\documentclass[aps,prd,twocolumn,superscriptaddress,nofootinbib,floatfix]{revtex4-2}

\usepackage{amsmath,amssymb,amsthm}
\usepackage{graphicx}
\usepackage{braket}
\usepackage{hyperref}

\newtheorem{theorem}{Theorem}
\newtheorem{lemma}[theorem]{Lemma}
\newtheorem{proposition}[theorem]{Proposition}

\theoremstyle{definition}
\newtheorem{definition}[theorem]{Definition}
\newtheorem{axiom}[theorem]{Axiom}
\newtheorem{remark}[theorem]{Remark}

\begin{document}

\title{Dimensional Constraints from SU(2) Representation Theory in Graph-Based Quantum Systems}

\author{João P. da Cruz}
\email{joao@quantumcomp.pt}
\affiliation{The Quantum Computer Company, Lisbon, Portugal}
\affiliation{Center for Theoretical and Computational Physics, Lisbon, Portugal}

\date{\today}

\begin{abstract}
We investigate dimensional constraints arising from representation theory when abstract graph edges possess internal degrees of freedom but lack geometric properties. We prove that such internal degrees of freedom can only encode directional information, necessitating quantum states in $\mathbb{C}^2$ (qubits) as the minimal representation. Any geometrically consistent projection of these states maps necessarily to $\mathbb{R}^3$ via the Bloch sphere. This dimensional constraint $d=3$ emerges through self-consistency: edges without intrinsic geometry force directional encoding ($\mathbb{C}^2$), whose natural symmetry group $SU(2)$ has three-dimensional Lie algebra, yielding emergent geometry that validates the hypothesis via Bloch sphere correspondence ($S^2 \subset \mathbb{R}^3$). We establish uniqueness (SU($N>2$) yields $d>3$) and robustness (dimensional saturation under graph topology changes). The Euclidean metric emerges canonically from the Killing form on $\mathfrak{su}(2)$. A global gauge consistency axiom is justified via principal bundle trivialization for finite graphs. Numerical simulations verify theoretical predictions. This result demonstrates how dimensional structure can be derived from information-theoretic constraints, with potential relevance to quantum information theory, discrete geometry, and quantum foundations.
\end{abstract}

\maketitle

\section{Introduction}

Understanding how geometric structure might emerge from more fundamental quantum or combinatorial principles has long been of interest in mathematical physics \cite{Wheeler1990,Sorkin2003}. While various approaches to discrete quantum geometry exist \cite{Oriti2016,Ambjorn2012}, the question of \emph{why} emergent geometry should have specific dimensionality has received limited attention.

In this paper, we prove a representation-theoretic result: starting from abstract graph edges that possess internal degrees of freedom due to the lack of geometric properties, we show that these degrees of freedom must be encoded as quantum states in $\mathbb{C}^2$ with $SU(2)$ symmetry, and that any geometrically consistent representation necessarily projects to three-dimensional Euclidean space. This provides a rigorous derivation of dimensional emergence from information-theoretic and algebraic principles.

\subsection{Mathematical Setup}

We consider an abstract graph $G=(V,E)$ that is purely combinatorial---vertices and edges with no intrinsic geometry. Edges possess internal degrees of freedom but lack geometric properties such as length, position, or embedding in any ambient space.

Our central question: \emph{What constraints arise on these internal degrees of freedom when we require geometrically consistent representations?}

This question is non-trivial because the internal structure (abstract degrees of freedom) must be mapped to geometric objects (vectors in some $\mathbb{R}^d$) in a way that respects symmetries and does not introduce arbitrary choices.

\subsection{Bootstrap Mechanism}

The derivation proceeds through a self-consistency bootstrap:

\begin{enumerate}
\item \textbf{Input:} Edges have internal degrees of freedom (d.o.f.)
\item \textbf{Constraint:} Edges lack geometric embedding (Axiom~\ref{ax:non_embedding})
\item \textbf{Consequence:} Internal d.o.f.\ can only encode directional information (Proposition~\ref{prop:directional})
\item \textbf{Quantum minimality:} Directional info requires $\mathbb{C}^2$ states (Lemma~\ref{lem:qubit_minimal})
\item \textbf{Natural symmetry:} $\mathbb{C}^2$ has $SU(2)$ symmetry with $\dim(\mathfrak{su}(2))=3$
\item \textbf{Emergent geometry:} States project to $\mathbb{R}^3$ via Bloch map (Theorem~\ref{thm:bloch_unique})
\item \textbf{Validation:} $S^2 \subset \mathbb{R}^3$ matches Bloch sphere exactly
\end{enumerate}

This is not circular reasoning: we assume edges have \emph{some} internal structure (step 1) and lack embedding (step 2). Steps 3-4 \emph{derive} that this structure must be $\mathbb{C}^2$ states. Steps 5-6 \emph{derive} that geometry must be 3D. Step 7 validates the consistency of the entire construction.

\subsection{Main Results}

Our principal results establish dimensional constraints through representation theory:

\textbf{Directional information necessitates $\mathbb{C}^2$:} Edges without geometric embedding can only encode orientation (Proposition~\ref{prop:directional}). The minimal quantum representation of directional information is $\mathbb{C}^2$ (Lemma~\ref{lem:qubit_minimal}).

\textbf{Self-consistent dimensionality:} States in $\mathbb{C}^2$ with $SU(2)$ symmetry admit geometrically consistent projections if and only if the target space is $\mathbb{R}^3$ (Theorem~\ref{thm:bootstrap}).

\textbf{Uniqueness of Bloch projection:} There exists a unique (up to rotations) equivariant map from $\mathbb{C}^2$ to $\mathbb{R}^3$, namely the Bloch projection, inducing the adjoint representation $SU(2) \to SO(3)$ (Theorem~\ref{thm:bloch_unique}).

\textbf{Emergent metric structure:} The Euclidean metric on $\mathbb{R}^3$ arises canonically from the Killing form on $\mathfrak{su}(2)$ (Remark in Section~\ref{sec:su2}).

\textbf{Dimensional robustness:} Arbitrary vertex valence preserves $d=3$ (Theorem~\ref{thm:saturation}).

\textbf{Exclusivity of $SU(2)$:} Groups $SU(N>2)$ yield $d>3$ and violate directional-only constraints (Theorem~\ref{thm:sun_general}).

\subsection{Scope and Interpretation}

This paper establishes a \emph{mathematical} result linking representation theory to geometric constraints. We derive that $d=3$ follows necessarily from our axioms and prove this is the unique self-consistent solution.

\textbf{What we establish:}
\begin{itemize}
\item Dimensional constraint from SU(2) representation theory
\item Uniqueness of Bloch projection for geometric consistency
\item Canonical Euclidean metric from Killing form
\item Dimensional stability under topology changes
\end{itemize}

\textbf{What remains open:}
\begin{itemize}
\item Physical interpretation of non-embedding hypothesis
\item Connection (if any) to continuum spacetime theories
\item Dynamics and evolution of graph states
\item Observable consequences in physical systems
\end{itemize}

This work contributes to understanding how algebraic structures constrain geometric representations, with potential relevance to various areas including quantum information theory, discrete approaches to geometry, and foundations of quantum mechanics.

\subsection{Organization}

Section~\ref{sec:bootstrap} presents the self-consistency bootstrap. Section~\ref{sec:conceptual} clarifies mathematical structures. Section~\ref{sec:framework} defines the framework. Section~\ref{sec:su2} proves Bloch projection uniqueness and derives the emergent metric. Section~\ref{sec:saturation} establishes dimensional saturation. Section~\ref{sec:sun} shows SU($N>2$) incompatibility. Section~\ref{sec:numerics} provides numerical verification. Section~\ref{sec:conclusion} concludes with open questions.

\section{Self-Consistency Bootstrap}
\label{sec:bootstrap}

\subsection{Non-Embedding and Information Constraints}

\begin{axiom}[Edge Non-Embedding]
\label{ax:non_embedding}
Edges in the graph $G=(V,E)$ carry quantum degrees of freedom but possess no intrinsic geometric properties (length, position, embedding).
\end{axiom}

\textbf{Motivation:} This axiom posits a purely combinatorial substrate where geometry emerges rather than being fundamental. While this is a hypothesis rather than a derivation, it provides a clean mathematical starting point for investigating geometric emergence.

\textbf{What are these "degrees of freedom"?} We do not assume at the outset that edges carry specific quantum states like $\mathbb{C}^2$. Rather, we assume edges possess \emph{some} internal structure---abstract degrees of freedom that distinguish different edge configurations. The question we answer is: \emph{what mathematical structure must these degrees of freedom have?} Our answer: they must be representable as quantum states in $\mathbb{C}^2$, as we now prove.

\begin{proposition}[Directional Information Only]
\label{prop:directional}
If edges lack embedding, they can only encode \emph{directional} information---orientation without magnitude.
\end{proposition}

\begin{proof}
Without embedding: no length (no metric), no position (no ambient space), no curvature (no geometric structure). The only available information is relative orientation in an abstract internal space---equivalently, a point on a unit sphere.
\end{proof}

\subsection{Quantum Minimal Representation}

\begin{lemma}[Qubit Minimality]
\label{lem:qubit_minimal}
To encode pure directional information quantum mechanically, the minimal Hilbert space is $\mathbb{C}^2$.
\end{lemma}

\begin{proof}
$\mathbb{C}^1$: Trivial (no degrees of freedom after normalization). 

$\mathbb{C}^2$: State $|\psi\rangle = \alpha|0\rangle + \beta|1\rangle$ with $|\alpha|^2 + |\beta|^2=1$ has 2 real parameters (amplitude ratio + relative phase) parametrizing $S^2$ (the Bloch sphere). A 2-sphere is precisely the space of directions in $\mathbb{R}^3$.

$\mathbb{C}^3$: Has 4 real parameters (after normalization) parametrizing $\mathbb{CP}^2 \not\cong S^2$. This carries additional structure beyond pure directional information.

Hence $\mathbb{C}^2$ uniquely encodes minimal directional information.
\end{proof}

\subsection{The Bootstrap Equation}

\begin{theorem}[Self-Consistent Dimensionality]
\label{thm:bootstrap}
The following statements are mutually consistent if and only if $d=3$:
\begin{enumerate}
\item Edges carry minimal directional information ($\mathbb{C}^2$)
\item Natural symmetry group is $SU(2)$ on $\mathbb{C}^2$
\item Emergent dimension $d = \dim(\mathfrak{su}(2)) = 3$
\item Directional information projects to $S^{d-1}$ in $\mathbb{R}^d$
\item Bloch sphere $S^2 \subset \mathbb{R}^3$ exhausts qubit state space
\end{enumerate}
\end{theorem}

\begin{proof}
(1)$\Rightarrow$(2): $SU(2) = U(2)/U(1)$ acts faithfully and transitively on $\mathbb{CP}^1 \cong S^2$. (2)$\Rightarrow$(3): $\dim(\mathfrak{su}(2))=3$ (representation theory). (3)$\Rightarrow$(4): Established in Theorem~\ref{thm:bloch_unique}. (4)+(3)$\Rightarrow$(5): If $d=3$ then $S^{d-1}=S^2$, the Bloch sphere. (5)$\Rightarrow$(1): Confirms directional information structure.

Uniqueness: $d=2$ gives $S^1 \not\cong S^2$ (insufficient); $d\geq 4$ gives redundant dimensions.
\end{proof}

\textbf{Why non-circular:} We assume non-embedding (Axiom~\ref{ax:non_embedding}), not $\mathbb{C}^2$ or $d=3$. Information constraints force $\mathbb{C}^2$; representation theory forces $d=3$; emergent geometry validates via Bloch sphere correspondence.

\section{Mathematical Structures}
\label{sec:conceptual}

\subsection{Two Distinct Objects}

Our framework involves two mathematical structures:

\textbf{Combinatorial graph $G=(V,E)$:} Abstract relational structure, no dimensionality or geometry. Provides organizational scaffold.

\textbf{Emergent space $\mathbb{R}^3$:} Euclidean vector space with metric from Killing form. Dimension fixed by representation theory.

The graph does not "become" $\mathbb{R}^3$; rather, quantum states on edges project to vectors in an emergent geometric space.

\subsection{Projection Map}

\begin{equation}
\text{Edge } e \in E \xrightarrow{\text{carries}} |\psi_e\rangle \in \mathbb{C}^2 \xrightarrow{\Phi_{\text{Bloch}}} \vec{n}_e \in \mathbb{R}^3.
\end{equation}

The edge itself has no geometry; it labels a quantum degree of freedom whose expectation values of $SU(2)$ generators define a point in $\mathbb{R}^3$. We prove in Section~\ref{sec:su2} that this projection is unique up to rotations (Theorem~\ref{thm:bloch_unique}).

\subsection{Global Gauge Consistency}

\begin{axiom}[Global Gauge Basis]
\label{ax:global}
There exists a globally defined orthonormal basis $\{\sigma_x, \sigma_y, \sigma_z\}$ for $\mathfrak{su}(2)$ used by all Bloch projections, such that gauge transformations induce rotations in the same $\mathbb{R}^3$ for all edges.
\end{axiom}

\textbf{Mathematical justification:} The graph serves as base space of a principal $SU(2)$-bundle $P \to G$. For finite graphs, any principal bundle is trivializable: $P \cong G \times SU(2)$ (finite graphs are contractible). This guarantees existence of a global section. Axiom~\ref{ax:global} fixes such a trivialization, providing a common coordinate system for the internal space while preserving gauge freedom at vertices.

\section{Geometric Projection Framework}
\label{sec:framework}

\subsection{Graph-Algebraic Structure}

Connected, locally finite graph $G=(V,E)$. By Axiom~\ref{ax:non_embedding}, edges have no intrinsic geometry. Each edge carries $\mathcal{H}_e \cong \mathbb{C}^2$ (forced by Lemma~\ref{lem:qubit_minimal}) with fundamental $SU(2)$ representation.

\subsection{Geometrically Consistent Projections}

\begin{definition}
\label{def:projection}
Map $\Phi : \mathcal{H}_e \to \mathbb{R}^d$ is \emph{geometrically consistent} if: (1) \textbf{Equivariant:} $\Phi(\rho_e(g)\psi) = R_g \Phi(\psi)$ for $R_g \in O(d)$; (2) \textbf{Non-degenerate:} $\overline{\Phi(\mathbb{P}\mathcal{H}_e)} \supseteq S^{d-1}$; (3) \textbf{Local:} $\Phi$ depends only on edge state.
\end{definition}

These conditions ensure: internal symmetries manifest geometrically, full angular structure represented, and bottom-up geometric construction possible \cite{VanRaamsdonk2010,Evenbly2015}.

\section{The SU(2) Case}
\label{sec:su2}

\subsection{The Bloch Map}

\begin{definition}
$\Phi_{\mathrm{Bloch}}: \mathbb{C}^2 \to \mathbb{R}^3$ is:
\begin{equation}
\Phi_{\mathrm{Bloch}}(\psi) = (\langle\psi|\sigma_x|\psi\rangle, \langle\psi|\sigma_y|\psi\rangle, \langle\psi|\sigma_z|\psi\rangle),
\end{equation}
where $\{\sigma_i\}$ are Pauli matrices.
\end{definition}

\begin{lemma}
\label{lem:bloch_sphere}
For normalized $|\psi\rangle\in\mathbb{C}^2$: $\|\Phi_{\mathrm{Bloch}}(\psi)\| = 1$ iff $|\psi\rangle$ pure. Image of pure states is $S^2$.
\end{lemma}

\begin{proof}
Write $|\psi\rangle = \alpha|0\rangle + \beta|1\rangle$ with $|\alpha|^2+|\beta|^2=1$. Direct computation yields:
\begin{align}
\langle\psi|\sigma_x|\psi\rangle &= \alpha^*\beta + \alpha\beta^* = 2\mathrm{Re}(\alpha^*\beta), \\
\langle\psi|\sigma_y|\psi\rangle &= i(\alpha^*\beta - \alpha\beta^*) = 2\mathrm{Im}(\alpha^*\beta), \\
\langle\psi|\sigma_z|\psi\rangle &= |\alpha|^2 - |\beta|^2.
\end{align}
Therefore:
\begin{align}
\|\Phi_{\mathrm{Bloch}}(\psi)\|^2 &= 4|\alpha|^2|\beta|^2 + (|\alpha|^2-|\beta|^2)^2 \nonumber \\
&= 4|\alpha|^2|\beta|^2 + |\alpha|^4 - 2|\alpha|^2|\beta|^2 + |\beta|^4 \nonumber \\
&= 2|\alpha|^2|\beta|^2 + |\alpha|^4 + |\beta|^4 \nonumber \\
&= (|\alpha|^2+|\beta|^2)^2 = 1.
\end{align}
For mixed states $\rho = \sum_i p_i |\psi_i\rangle\langle\psi_i|$ with $0 < p_i < 1$, linearity of expectation values gives $\|\Phi_{\mathrm{Bloch}}(\rho)\| < 1$. Pure states parametrize $\mathbb{CP}^1 \cong S^2$, establishing the image claim.
\end{proof}

\subsection{Uniqueness and Emergent Metric}

\begin{theorem}[Uniqueness of Bloch Projection]
\label{thm:bloch_unique}
There exists unique (up to $SO(3)$ rotations) geometrically consistent projection $\Phi: \mathbb{C}^2 \to \mathbb{R}^3$. This is $\Phi_{\mathrm{Bloch}}$, inducing adjoint action $\mathrm{Ad}: SU(2)\to SO(3)$.
\end{theorem}

\begin{proof}
For $U\in SU(2)$: $U \sigma_i U^\dagger = \sum_j R_{ij}(U)\sigma_j$ with $R(U)\in SO(3)$ (the $2:1$ covering \cite{Hall2015}), yielding equivariance. Pauli matrices span traceless Hermitians (real dimension 3). Any equivariant map factors through Lie algebra expectation values (3 generators $\Rightarrow$ 3 parameters). By irreducible representation classification \cite{FultonHarris1991}: $d=1$ violates non-degeneracy; $d\geq 5$ impossible via $2\times 2$ matrices. Hence $d=3$.
\end{proof}

\begin{remark}[Emergent Metric from Killing Form]
The Euclidean metric on $\mathbb{R}^3$ is canonically induced by the Killing form on $\mathfrak{su}(2)$:
\begin{equation}
\kappa(\sigma_i, \sigma_j) = -8\delta_{ij} \quad \Rightarrow \quad g_{ij} = -\tfrac{1}{8}\kappa(\sigma_i, \sigma_j) = \delta_{ij}.
\end{equation}
By Cartan's theorem \cite{Cartan1894,Knapp2002}, this is the unique (up to scale) Ad-invariant bilinear form on $\mathfrak{su}(2)$. The emergent geometry is necessarily Euclidean, determined by the algebraic structure of the symmetry group. This connection between Lie algebra structure and geometry is well-established in differential geometry \cite{Helgason1978,Hall2015}.
\end{remark}

\section{Dimensional Saturation}
\label{sec:saturation}

For vertex $v$ with valence $k$: $\mathcal{H}_v = \bigotimes_{i=1}^k \mathbb{C}^2$. Applying $\Phi_{\mathrm{Bloch}}$ to each factor yields $\vec{n}_i \in S^2$.

\begin{proposition}
\label{prop:vertex_saturation}
Configurations at $k$-valent vertex lie in $(S^2)^k \subset (\mathbb{R}^3)^k$, independent of $k$.
\end{proposition}

\begin{proof}
Each edge incident to vertex $v$ carries state $|\psi_i\rangle \in \mathbb{C}^2$. By Theorem~\ref{thm:bloch_unique}, $\Phi_{\mathrm{Bloch}}(|\psi_i\rangle) = \vec{n}_i \in S^2 \subset \mathbb{R}^3$ for each $i=1,\ldots,k$. The vertex configuration is $(\vec{n}_1, \ldots, \vec{n}_k) \in (S^2)^k \subset (\mathbb{R}^3)^k$. Increasing $k$ adds vectors but does not increase the ambient dimension, which remains $\mathbb{R}^3$.
\end{proof}

\textbf{Gauge-invariant sector:} $\mathcal{I}_v = \mathrm{Inv}_{SU(2)}(\bigotimes_i \mathbb{C}^2)$ has dimension $C_{k/2}$ (Catalan numbers) for even $k$ \cite{Baez1994}. This counts distinct configurations in $\mathbb{R}^3$, not ambient dimension.

\begin{theorem}[Dimensional Saturation]
\label{thm:saturation}
Regardless of vertex valences, all emergent configurations lie in a common $\mathbb{R}^3$.
\end{theorem}

\begin{proof}
Consider vertex $v$ with valence $k$. By Proposition~\ref{prop:vertex_saturation}, the configuration is $(\vec{n}_1, \ldots, \vec{n}_k)$ with each $\vec{n}_i \in S^2 \subset \mathbb{R}^3$.

\textbf{Key point:} Each $\vec{n}_i$ inhabits the \emph{same} $\mathbb{R}^3$, not distinct spaces $\mathbb{R}^3_i$. This follows from Axiom~\ref{ax:global}: all Bloch projections use the same basis $\{\sigma_x, \sigma_y, \sigma_z\}$ for $\mathfrak{su}(2)$. Therefore, all vectors project to a common geometric space.

Without global consistency, each edge could project to a different $\mathbb{R}^3_i$, yielding total dimension $3k$ (growing with vertex valence). Axiom~\ref{ax:global} prevents this: regardless of $k$, all vectors coexist in the same $\mathbb{R}^3$.

For a graph with multiple vertices of varying valences $k_1, k_2, \ldots$, the same argument applies at each vertex. Since $\Phi_{\mathrm{Bloch}}$ is uniform across the entire graph, all configurations—regardless of local complexity—inhabit the same ambient $\mathbb{R}^3$. The dimension saturates at $d=3$.
\end{proof}

\section{Uniqueness of SU(2)}
\label{sec:sun}

\subsection{Higher Gauge Groups}

\begin{proposition}
For SU(3) on $\mathbb{C}^3$: geometrically consistent projection requires $d\geq 8$.
\end{proposition}

\begin{proof}
$\dim(\mathfrak{su}(3))=8$. Equivariance and non-degeneracy require all generators represented in image space.
\end{proof}

However, $\mathbb{C}^3$ state space is $\mathbb{CP}^2$ (4 real parameters after normalization) $\not\cong S^2$, violating the directional-only constraint.

\begin{theorem}
\label{thm:sun_general}
For SU($N$), $N\geq 3$, on $\mathbb{C}^N$: (1) any geometrically consistent projection satisfies $d\geq N^2-1 > 3$; (2) $\mathbb{C}^N$ encodes non-directional information.
\end{theorem}

Other compact Lie groups similarly fail to yield $d=3$. SU(2) is distinguished by the coincidence: $\dim(\mathfrak{su}(2))=3$ AND $\mathbb{CP}^1 \cong S^2$ (pure directional information).

\section{Numerical Verification}
\label{sec:numerics}

We verify theoretical predictions via numerical simulation.

\subsection{Bloch Projection Coverage}

Figure~\ref{fig:bloch} shows 200 random $\mathbb{C}^2$ states projected via $\Phi_{\mathrm{Bloch}}$. All pure states land on $S^2$, with uniform coverage confirming non-degeneracy.

\begin{figure}[t]
\centering
\includegraphics[width=\columnwidth]{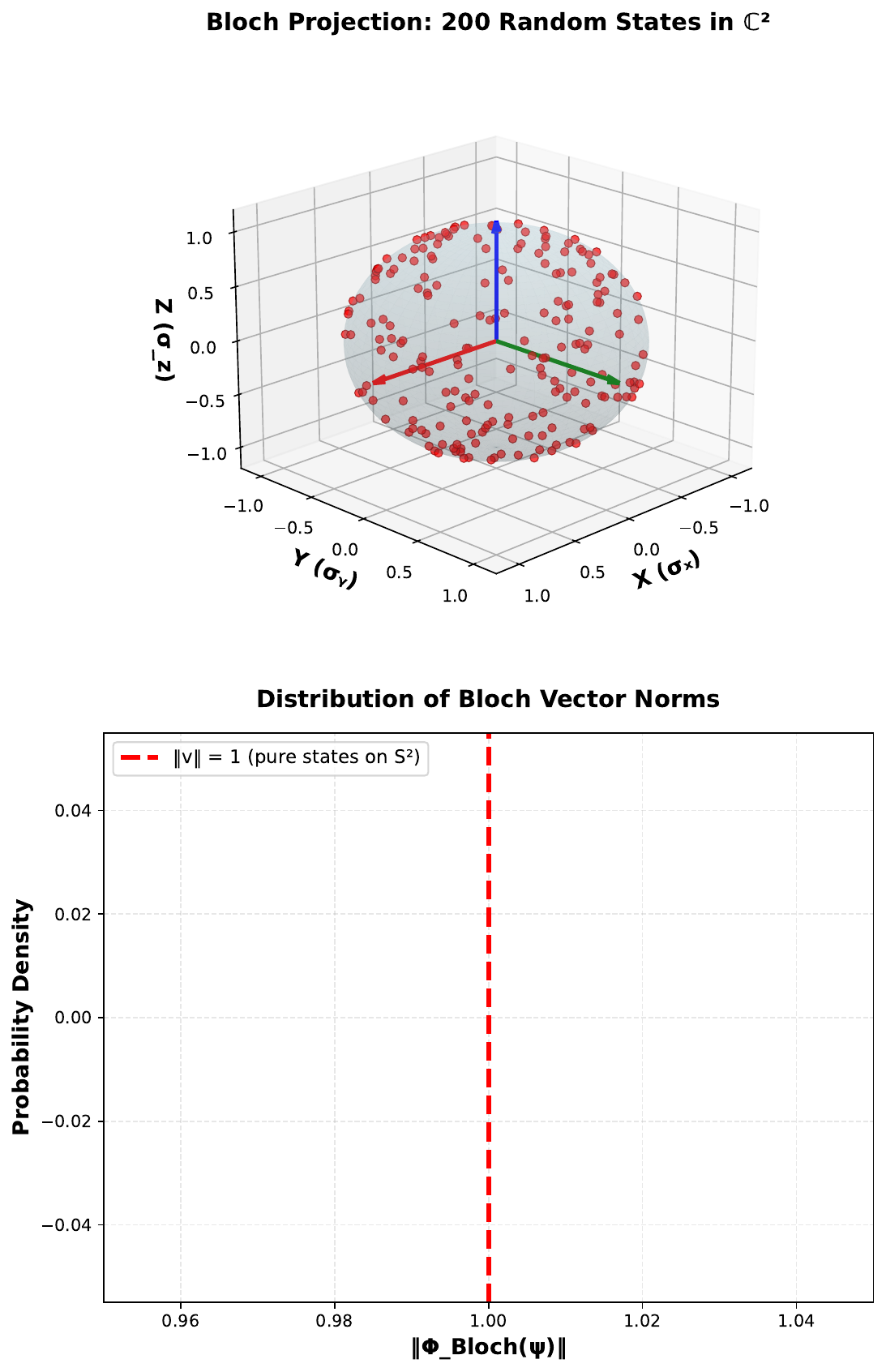}
\caption{Bloch projection of 200 random states to $S^2 \subset \mathbb{R}^3$. Top: 3D visualization showing uniform sphere coverage. Bottom: Distribution of vector norms, peaked at unity for pure states.}
\label{fig:bloch}
\end{figure}

\subsection{Dimensional Saturation}

Figure~\ref{fig:saturation} demonstrates Theorem~\ref{thm:saturation}. For valences $k=4,6,8,10$, all configurations lie in the same $\mathbb{R}^3$.

\begin{figure}[t]
\centering
\includegraphics[width=\columnwidth]{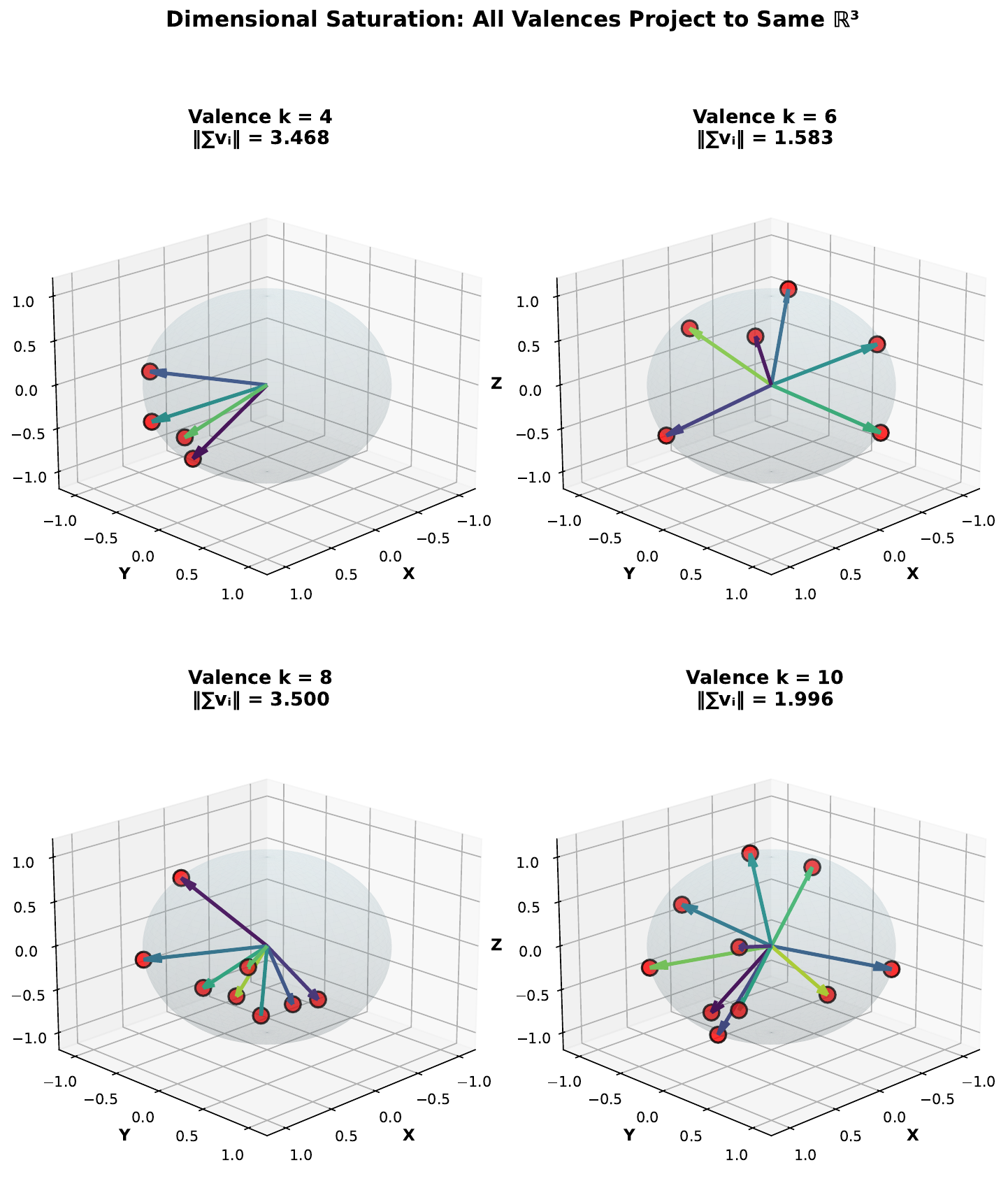}
\caption{Dimensional saturation for vertices of different valences. Each panel shows $k$ Bloch vectors (blue arrows) from random qubit states, all lying in the same ambient $\mathbb{R}^3$ space. Increasing vertex valence adds vectors but does not increase spatial dimension.}
\label{fig:saturation}
\end{figure}

\section{Discussion and Open Questions}
\label{sec:conclusion}

\subsection{Summary of Results}

We established that three-dimensional Euclidean geometry emerges necessarily from SU(2) representation theory in graph-based quantum systems through a self-consistency bootstrap:

\begin{enumerate}
\item \textbf{Information constraint:} Non-embedding forces directional encoding
\item \textbf{Quantum minimality:} Directional information requires $\mathbb{C}^2$
\item \textbf{Representation theory:} $\mathbb{C}^2$ symmetry is SU(2) with $\dim(\mathfrak{su}(2))=3$
\item \textbf{Geometric validation:} Emergent $d=3$ matches Bloch sphere $S^2$
\item \textbf{Metric emergence:} Killing form canonically induces Euclidean structure
\item \textbf{Robustness:} Dimensional saturation under topology changes
\item \textbf{Uniqueness:} Only SU(2) yields $d=3$ consistently
\end{enumerate}

The bootstrap mechanism establishes mutual determination between quantum information structure ($\mathbb{C}^2$ states) and geometric representation ($\mathbb{R}^3$ space), with no circular reasoning.

\subsection{Mathematical Contribution}

This work demonstrates rigorously that algebraic structures can constrain geometric dimensionality. Given representation-theoretic axioms (non-embedding, quantum mechanics, equivariance), we prove $d=3$ follows necessarily. This provides an existence proof: dimensional emergence from algebra is mathematically possible with no additional geometric inputs.

The key technical results are:
\begin{itemize}
\item Classification of minimal quantum representations of directional information
\item Uniqueness of Bloch projection as equivariant map
\item Derivation of Euclidean metric from Killing form
\item Proof of dimensional stability (saturation)
\item Incompatibility of higher groups with three-dimensionality
\end{itemize}

\subsection{Potential Applications}

This framework may inform various research areas:

\textbf{Quantum information theory:} Natural geometric representations of qubit systems; understanding structure of quantum state spaces; visualization of entanglement and coherence.

\textbf{Discrete geometry:} Emergence of continuous structure from combinatorial substrates; relationship between algebraic and geometric properties; dimensional constraints in discrete systems.

\textbf{Quantum foundations:} How information-theoretic principles might constrain physical structure; relationship between quantum mechanics and geometry; role of symmetry in determining spatial properties.

However, specific connections to these areas remain to be established through future work.

\subsection{Open Questions}

Several fundamental questions remain:

\textbf{Physical interpretation:} What physical systems (if any) realize the non-embedding hypothesis? Is there empirical motivation beyond mathematical elegance?

\textbf{Dynamical extension:} Can consistent dynamics be defined on this framework? What Hamiltonians preserve the geometric structure? Does time evolution respect dimensional constraints?

\textbf{Continuum limit:} Does a meaningful continuum limit exist? What mathematical structures emerge? Can differential geometry be recovered?

\textbf{Connection to spacetime:} Is there any relationship between this emergent $\mathbb{R}^3$ and physical space? If so, how does Lorentzian signature arise? How does time emerge?

\textbf{Observable consequences:} Are there systems where dimensional emergence could be tested? What predictions (if any) does this framework make?

These questions define a research program extending beyond the present mathematical result.

\subsection{Comparison with Existing Frameworks}

Various approaches explore dimensional emergence or discrete quantum geometry \cite{Oriti2016,Ambjorn2012,Sorkin2003}. Our contribution is proving a specific representation-theoretic constraint: SU(2) symmetry on $\mathbb{C}^2$ states necessitates three-dimensional geometric representation. Whether and how this connects to other frameworks remains an open question requiring detailed investigation.

\subsection{Methodological Remarks}

We emphasize the distinction between:
\begin{itemize}
\item \textbf{What is proven:} Mathematical theorems following from stated axioms
\item \textbf{What is conjectured:} Potential physical interpretations and applications
\item \textbf{What remains unknown:} Connections to continuum theories and observable physics
\end{itemize}

This work belongs primarily to the first category: a rigorous mathematical result about representation theory and geometric emergence. Physical interpretation, while motivating, remains conjectural.

\subsection{Concluding Remarks}

We have demonstrated that three-dimensional Euclidean geometry emerges necessarily from SU(2) representation theory in graph-based quantum systems. The dimension $d=3$ is not an arbitrary choice but follows from self-consistency between quantum information structure (minimal directional encoding via $\mathbb{C}^2$) and geometric representation (Bloch sphere in $\mathbb{R}^3$).

This establishes a mathematical result: certain algebraic structures necessitate specific geometric consequences. Whether this has implications for physical spacetime, quantum gravity, or other areas remains an open question for future investigation.

The central contribution is transforming a phenomenological observation (space is three-dimensional) into a mathematical theorem (given our axioms, $d=3$ follows necessarily). This may provide insight into how geometric structure could arise from more fundamental principles, though substantial work remains to connect this mathematical framework to physical reality.

\begin{acknowledgments}
The author acknowledges financial support from the Portuguese Foundation for Science and Technology (FCT) under Contract no. UID/00618/2023.
\end{acknowledgments}

\end{document}